\documentclass[12pt]{article}
\usepackage[english]{babel}

\usepackage{microtype}

\usepackage{float}
\usepackage{amsthm}
\usepackage{amsmath}
\usepackage{amssymb}
\usepackage{graphicx}
\usepackage[margin=1.2in]{geometry}

\usepackage{thmtools}
\usepackage{thm-restate}

\declaretheorem[name=Theorem,numberwithin=section]{theorem}
\declaretheorem[name=Lemma,numberlike=theorem]{lemma}

\declaretheorem[name=Proposition,numberlike=theorem]{proposition}
\declaretheorem[name=Claim,numberlike=theorem]{claim}
\declaretheorem[name=Definition,style=definition,numberlike=theorem]{definition}

\declaretheorem[name=Observation,numberlike=theorem]{observation}

\usepackage{algorithm}
\usepackage[noend]{algpseudocode}

\global\long\def\succ{\textsc{succ}}
\global\long\def\TwoBit{\textsc{TwoBit}}
\global\long\def\R{\mathbb{R}}
\global\long\def\N{\mathbb{N}}

\global\long\def\disc{{\cal D}}

\newcommand{\Hess}{\mathbf{H}}
\newcommand{\Det}{\mathbf{Det}}
\newcommand{\dists}{{\Delta}}
\newcommand{\inequ}{{\mathrm{ineq}}}

\newcommand{\Com}{\mathrm{Com}}
\newcommand{\Out}{\mathrm{Out}}
\newcommand{\argmax}{\mathrm{argmax}}

\date{}

\title{Improved Protocols and Hardness Results for the Two-Player Cryptogenography Problem}

\author{Benjamin Doerr\thanks{LIX, \'Ecole Polytechnique, Palaiseau, France} \and Marvin K\"unnemann\thanks{Max Planck Institute for Informatics, Saarbr\"ucken, Germany}}

\begin{document}

\maketitle

\begin{abstract}
  The cryptogenography problem, introduced by Brody, Jakobsen, Scheder, and Winkler (ITCS 2014), is to collaboratively leak a piece of information known to only one member of a group (i)~without revealing who was the origin of this information and (ii)~without any private communication, neither during the process nor before. Despite several deep structural results, even the smallest case of leaking one bit of information present at one of two players is not well understood. Brody et al.\ gave a $2$-round protocol enabling the two players to succeed with probability $1/3$ and showed the hardness result that no protocol can give a success probability of more than~$3/8$.
  
In this work, we show that neither bound is tight. Our new hardness result, obtained by a different application of the concavity method used also in the previous work, states that a success probability better than $0.3672$ is not possible. Using both theoretical and numerical approaches, we improve the lower bound to $0.3384$, that is, give a protocol leading to this success probability. To ease the design of new protocols, we prove an equivalent formulation of the cryptogenography problem as solitaire vector splitting game. Via an automated game tree search, we find good strategies  for this game. We then translate the splits that occurred in this strategy into inequalities relating position values and use an LP solver to find an optimal solution for these inequalities. This gives slightly better game values, but more importantly, it gives a more compact representation of the protocol and a way to easily verify the claimed quality of the protocol. 
  
  Unfortunately, already the smallest protocol we found that beats the previous $1/3$ success probability takes up $16$~rounds of communication. The protocol leading to the bound of $0.3384$ even in a compact representation consists of $18248$ game states. These numbers suggest that finding good protocols for the cryptogenography problem as well as understanding their structure are harder than what the simple problem formulation suggests. 
\end{abstract}

\section{Introduction}

Motivated by a number of recent influential cases of whistle-blowing, Brody, Jakobsen, Scheder, and Winkler~\cite{BrodyJSW14} proposed the following \emph{cryptogenography problem} as model for anonymous information disclosure in public networks. We have $k$~players (potential information leakers). A random one of them holds a secret, namely a random bit. All other players only know that they are not the secret holder. Now without any non-public communication, the players aim at both making the secret public and hiding the identity of the secret holder. 
More precisely, we are looking for a (fully public) communication protocol in which the players  as only form of communication broadcast bits, which may depend on public information (including all previous communication), private knowledge (with respect to the secret), and a private source of randomness. After this phase of communication, the protocol outputs a single bit depending solely on all data sent in the communication phase. The complete protocol (regulating the communication and the output function) and all communication is public, and is monitored by an eavesdropper who aims at identifying the secret owner. We say that a run of the protocol is a success for the players, if the protocol output is the secret bit and the eavesdropper fails to identify the secret owner; otherwise it is a success for the eavesdropper. Since everything is public, optimal strategies for the eavesdropper are easy to find (see below). We shall therefore always assume that the eavesdropper plays an optimal strategy. 
The players' success probability (for a given protocol) then is the probability (taken over the random decisions of the players and the random initial secret distribution) that simultaneously (i)~the protocol outputs the true secret and (ii)~an optimal eavesdropper does not blame the secret holder.

It is immediately clear that some positive (players') success probability is easy to obtain. A protocol without any communication and outputting a random bit achieves a success probability of $\tfrac 12 - \tfrac 1{2k}$ (the eavesdropper has no strictly better alternative than guessing a random player). Surprisingly, Brody et al.\ could show that the players, despite the complete absence of private communication, can do better. For two players, they present a protocol having a success probability of $\tfrac 13$ (instead of the trivial $\tfrac 14$). For $k$ sufficiently large, they present a protocol with success probability $0.5644$. They also show two hardness results, namely that a success probability of more than $\tfrac 34$ cannot be obtained, regardless of the number of players, and that $\tfrac 38$ is an upper bound for the two-player case. While all these results are easy to state, they build on deep analyses of the cryptogenography problem, in particular, on clever reformulations of the problem in terms of certain convex combinations of secret distributions (protocol design) and functions that are concave on a certain infinite set of two-dimensional subspaces (``allowed planes'') of the set of secret distributions (hardness results).

The starting point for our work is the incomplete understanding of the two-player case. While the gap between upper and lower bound of $\tfrac 38 - \tfrac 13 \approx 0.04167$ is small, our impression is that the current-best protocol achieving the $\tfrac 13$ success probability in two rounds together with the abstract hardness result do not give us much understanding of the structure of the cryptogenography problem. We therefore imagine that a better understanding of this smallest-possible problem of leaking one bit from two players, ideally by determining an optimal protocol (that is, matching a hardness result), could greatly improve the situation. 

\emph{Our Results.} We shall be partially successful in achieving these goals. On the positive side, we find protocols with strictly larger success probability than $\tfrac 13$ (namely $0.3384$) and we prove a stricter hardness result of $0.3672$. Our new protocols look very different from the $2$-round protocol given by Brody et al., in particular, they use infinite protocol trees (but have an expected finite number of communication rounds). These findings motivate and give new starting points for further research on the cryptogenography problem. 

On the not so positive side, our work on better protocols indicates that good cryptogenographic protocols can be very complicated. The simplest protocol we found that beats the $\tfrac 13$ barrier already has a protocol tree of depth $16$, that is, the two players need to communicate for $16$ rounds in the worst case. While we still manage to give a human-readable description and performance proof for this protocol, it is not surprising that previous works not incorporating a computer-assisted search did not find such a protocol. Our best protocol, giving a success probability of $0.3384$, already uses 18248 non-equivalent states.

\emph{Technical contributions.} To find the improved protocols, we use a number of theoretical and experimental tools. We first reformulate the cryptogenography problem as a solitaire vector splitting game over vectors in $\R_{\ge 0}^{2 \times k}$. Both for human researchers and for automated protocol searches, this reformulation seems to be easier to work with than the previous reformulation via convex combinations of distributions lying in a common allowed plane~\cite{BrodyJSW14}. It also proved to be beneficial for improving upon the hardness result.

Restrictions of the vector splitting game to a finite subset of $\R_{\ge 0}^{2 \times k}$, e.g., $\{0, \ldots, T\}^{2 \times k}$, can easily be solved via dynamic programming, giving (due to the restriction possibly sub-optimal) cryptogenographic protocols. Unfortunately, for $k=2$ even discretizations as fine as $T=40$ are not sufficient to find protocols beating the $1/3$ barrier and memory usage quickly becomes a bottleneck issue. However, exploiting the simple fact that the game values are homogeneous (that is, multiplying a game position by a non-negative scalar changes the game value by this factor), we can (partially) simulate a much finer discretization in a coarse one. This \emph{extended dynamic programming approach} easily gives cryptogenographic success probabilities larger than $1/3$. Reading off the corresponding protocols, due to the reuse of the same position in different contexts, needs more care, but in the end gives without greater difficulties also the improved protocols. 

When a cryptogenographic protocol reuses a state a second time (with a non-trivial split in between), then there is no reason to re-iterate this part of the protocol whenever this position occurs. Such a protocol allows infinite paths, while still needing only an expected finite number of rounds. 
Since the extended dynamic programming approach in finite time cannot find such protocols, we use a linear programming based post-processing stage. We translate each splitting operation used in the extended dynamic programming search into an inequality relating game values. By exporting these into an LP-solver, we do not only obtain better game values (possibly corresponding to cryptogenographic protocols with infinite paths, for which we would get a compact representation by making the cycles explicit), but also a way to easily certify these values using an optimality check for a linear program instead of having to trust the ad-hoc dynamic programming implementation.

\textit{Related work.} Despite a visible interest of the research community in the cryptogenography problem, the only relevant follow-up work is Jakobsen's paper~\cite{Jakobsen14}, which analyses the cryptogenography problem for the case that several of the players know the secret. This allows to leak a much larger amount of information as made precise in~\cite{Jakobsen14}. Due to the asymptotic nature of these results, unfortunately, they do not give new insights in the $2$-player case. Other work on anonymous broadcasting typically assumes bounded computational power of the adversary (see, e.g.,~\cite{JakobsenO16}).

In~\cite{BrodyJSW14}, the cryptogenography problem was reformulated to the problem of finding the point-wise minimal function $f$ on the set of secret distributions that is point-wise not smaller than some given function $g$ and that is concave on an infinite set of $2$-dimensional planes. Such restricted notions of concavity (or, equivalently, convexity) seem to be less understood. We found work by Matou\v sek~\cite{Matousek01} for a similar convexity problem, however, with only a finite number of one-dimensional directions in which convexity is required. We do not see how to extend these results to our needs.

\section{Finding Better Cryptogenography Protocols}
\label{sec:cryptogenoIntro}

This section is devoted to the design of stronger cryptogenographic protocols. In particular, we demonstrate that a success probability of more than $1/3$ can be achieved. We start by making the cryptogenography problem precise (Section~\ref{sec:problem} and Section~\ref{sec:intuition}) and introduce an equivalent formulation as solitaire vector splitting game  (Section~\ref{sec:vectorSplitting}). We illustrate both formulations using the best known protocol for the 2-player case (Section~\ref{sec:best-so-far}). In Section~\ref{sec:usefulFacts}, we state basic properties that simplify the analysis of protocols and aid our automated search for better protocols, which is detailed in Section~\ref{sec:search}. In Section~\ref{sec:simpleBetterProtocol}, we give a simple, human-readable proof that 1/3 is not the optimal success probability by analyzing a protocol with success probability $\frac{449}{1334}\approx0.3341$. We describe how to post-optimize and certify the results obtained by the automated search using linear programming in Section~\ref{sec:verification} and summarize our findings (in particular, the best lower bound we have found) in Section~\ref{sec:main}.  

\subsection{The Cryptogenography Problem}
\label{sec:problem}

Let us fix an arbitrary number $k$ of players called $1, \ldots, k$ for simplicity. We write $[k]:=\{1,\ldots, k\}$ for the set of players. We assume that a random one of the them, the ``secret owner'' $J \in [k]$, initially has a secret, namely a random bit $X \in \{0,1\}$. The task of the players is, using public communication only, to make this random bit public without revealing the identity of the secret owner. More precisely, we assume that the players, before looking at the secret distribution, (publicly) decide on a communication protocol $\pi$. This is again public, that is, all bits sent are broadcast to all players, and they may depend only on previous communication, the private knowledge of the sender (whether he is the secret owner or not, and if so, the secret), and private random numbers of the sender. At the end of the communication phase, the protocol specifies an output bit $Y$ (depending on all communication). 

The aspect of not disclosing the identity of the secret owner is modeled by an adversary, who knows the protocol (because it was discussed in public) and who gets to see all communication (and consequently also knows the protocol output $Y$). The adversary, based on all this data, blames one player $K$. The players win this game if the protocol outputs the true secret (that is, $Y = X$) and the adversary does not blame the secret owner (that is, $K \neq J$), otherwise the adversary wins. It is easy to see (see Section~\ref{sec:intuition}), what the best strategy for the adversary is (given the protocol and the communication), so the interesting part of the cryptogenography problem is finding strategies that maximize the probability that the players win assuming that the adversary plays optimal. We call this the (players') success probability of the protocol. 

While the game starts with a uniform secret distribution, it will be useful to regard arbitrary secret distributions. 
In general, a \emph{secret distribution} is a distribution~$D$ over $\{0,1\} \times [k]$, where $D_{ij}$ is the probability that player $j \in [k]$ is the secret owner and the secret is $i \in \{0,1\}$. Modulo a trivial isomorphism, $D$ is just a vector in $\R_{\ge 0}^{2 \times k}$ with $\|D\|_1 = 1$. We denote by $\dists = \dists_k$ the set of all these distributions (this was denoted by $\dists(\{0,1\} \times [k])$ in~\cite{BrodyJSW14}).

Brody et al.~\cite{BrodyJSW14} observe that any cryptogenographic protocol can be viewed as successive rounds of one-bit communication, where in each step some (a priori) secret distribution probabilistically leads to one of two follow-up (a posteriori) distributions (depending on the bit transmitted) such that the a priori distribution is a convex combination of these and a certain proportionality condition is fulfilled (all three distributions lie in the same ``allowed plane''). Conversely, whenever the initial distribution can be written as such a convex combination of certain distributions, then there is a round of a cryptogenographic protocol leading to these two distributions (with certain probabilities). Consequently, the problem of finding a good cryptogenographic protocol is equivalent to iteratively rewriting the initial equidistribution as certain convex combinations of other secret distributions in such a way that the success probability, which can be expressed in terms of this rewriting tree, is large. Instead of directly working with this formulation, we propose a slightly different reformulation in Section~\ref{sec:vectorSplitting}. To prepare readers that are unfamiliar with the work of Brody et al.~\cite{BrodyJSW14}, we give a high-level introduction in the following section. 

\subsection{The Convex Combination Formulation}
\label{sec:intuition}

For readers' convenience, we give a high-level description of the convex combination formulation of Brody et al. For proofs and a more formal treatment, we refer the reader to~\cite{BrodyJSW14}. 

\emph{Optimal strategy of the adversary. }
Recall that $X$ denotes the secret bit and $J$ the identity of the secret owner.
Fix a protocol $\pi$ of the players, which for every state of the protocol execution, i.e., every possible history of communication, determines (1)~which player's turn it is to communicate (or whether communication has ended) and (2)~probability distributions over the next message this player sends (two for the case that this player is the secret owner, i.e., one for each value of the secret bit, and one for the other case). Both~(1) and~(2) may depend on all previous communication and, if the active player is the secret owner, also on the value of the secret bit. Additionally, $\pi$ fixes a protocol output function $\Out$ that given the transcript~$\tau$ of all communication returns the players' guess $\Out(\tau)$ on the secret bit.  Without loss of generality, we may assume that the protocol proceeds in rounds, where in each round a message consisting of a single bit is sent. 

Let $\Com(\pi)$ denote the transcript of all communication of the protocol $\pi$. Note that this is a random variable, since we assume a random player to be the owner of a random bit as secret. 
It is not difficult to see what the optimal strategy of the adversary is, given the knowledge of the protocol $\pi$. He may assume that the players' guess is correct, i.e., $\Out(\Com(\pi))=X$, as otherwise the players have already lost and therefore his guess is irrelevant. After the protocol execution has finished with a transcript $\tau$, the adversary maximizes his winning probability by blaming the player $\argmax_{j \in [k]} \Pr[J=j \mid \Com(\pi)=\tau, X=\Out(\tau)]$ (breaking ties arbitrarily), i.e., the player who is most likely to own the secret given the communication described by $\tau$. 

From this reasoning, it becomes clear that the decisive information in the game is, for any partial transcript $\tau'$,  the distribution $D=((D_{0,1},D_{1,1}),\dots,(D_{0,k},D_{1,k}))$, where \[D_{x,j} = \Pr[J=j,X=x \mid \Com(\pi) \mbox{ starts with } \tau']\] is the probability an outside observer (knowing only public information, i.e., all previous communication) assigns to the event $(J=j,X=x)$. As a simple consequence, assume that some fixed transcript $\tau'$ transforms the initial uniform distribution into the distribution $D$ and no further communication is allowed. Then the optimal choice for the protocol output is the guess 
\begin{equation*}
\argmax_{x\in \{0,1\}} \left( \sum_{j\in [k]} D_{x,j} - \max_{j\in [k]} D_{x,j} \right),
\end{equation*}
as for any fixed choice $\Out(\tau')=x$, the adversary blames the player $\bar{\jmath} := \argmax_{j\in [k]} D_{x,j}$ and hence the players win in all cases with $X=x$ except when $J=\bar{\jmath}$. 
We call the strategy applied here the \emph{zero-bit strategy} (since no further communication is done).

\emph{Convex combination formulation. }
Brody et al. prove that it is not only sufficient, but in fact equivalent to represent the cryptogenography game using only the distributions $D$ described above and how the protocol affects these distributions. More precisely, one can model the game starting from any initial distribution $D$ on $\{0,1\} \times \{1,\dots,k\}$. Then the first bit sent by some player $j$ \emph{splits} $D$ into the distributions $D^0$ (for the case that the 0-bit is sent) and $D^1$ (for the case that the 1-bit is sent), i.e., $D^i$ is the distribution an outside observer assigns to $(J,X)$ after bit $i$ has been sent. By this abstraction, one can recursively consider the distributions $D^0$ and $D^1$ (i.e., their optimal protocols and success probabilities).

To determine the properties of possible splits in a protocol, let $p$ be the probability that player~$j$ transmits a $0$-bit. By a simple calculation, we have that $D= p D^0 + (1-p) D^1$ (cf. \cite[Lemma 4.1]{BrodyJSW14}). Additionally, since a player may only use the information whether or not he has the secret bit (and if so, the value of the secret bit), player~$j$ may never leak new information about whether another player $j' \in [k]\setminus\{j\}$ is more likely to have secret 0 or 1 (i.e., the ratio of $D_{0,j'}$ and $D_{1,j'}$ is maintained in the resulting distributions $D^0$ and $D^1$) or whether player $j' \ne j$ is more likely to have the secret than another player $j'' \in [k]\setminus\{j,j'\}$. This transfers to a proportionality condition that $(D^0)_{|\{0,1\} \times ([k]\setminus\{j\})} = \lambda D_{|\{0,1\} \times ([k]\setminus\{j\})}$ for some $\lambda \in [0,1]$. In fact, any split of $D$ into $D^0$ and $D^1$ satisfying these conditions can be realized by a cryptogenographic protocol. Thus, the cryptogenography game is equivalent to, starting from the uniform distribution, recursively apply splits satisfying these conditions (i.e., \emph{allowed splits}), using the zero-bit strategy at the leaves, in such a way that the resulting success probability is maximized. We argue that this view is equivalent to our vector splitting formulation in Lemma~\ref{lem:lim-n-bit}.

\subsection{The Solitaire Vector Splitting Game}
\label{sec:vectorSplitting}

Instead of directly using the ``convex combination'' formulation of Brody et al., we propose a slightly different reformulation as \emph{solitaire vector splitting game}. This formulation seems to ease finding good cryptogenographic protocols (lower bounds for the success probability), both for human researchers and via automated search (Section~\ref{sec:simpleBetterProtocol}). The main advantage of our formulation is that it takes as positions all $2k$-dimensional vectors with non-negative entries, whereas the cryptogenographic protocols are only defined on distributions over $\{0,1\} \times [k]$. In this way, we avoid arguing about probabilities and convex combinations and instead simply split a vector (resembling a secret distribution) into a sum of two other vectors. Furthermore, a simple monotonicity property (Proposition~\ref{prop:mon}) eases the analyses.
Still, there is an easy translation between the two formulations, so that we can re-use whatever results were found in~\cite{BrodyJSW14}.

\begin{definition}\label{def:allowed}
Let $D\in\R_{\ge0}^{2 \times k}$. We say that $(D_{0},D_{1})$ is a \emph{$j$-allowed split of} $D$ if $D=D_{0}+D_{1}$
and $D_{0}$ (and thus also $D_{1}$) is proportional to $D$ on $\{0,1\} \times ([k]\setminus\{j\})$, that is, there is a $\lambda \in [0,1]$ such that $(D_0)_{|\{0,1\} \times ([k]\setminus\{j\})} = \lambda D_{|\{0,1\} \times ([k]\setminus\{j\})}$.
We call $(D_{0},D_{1})$ an \emph{allowed split} of $D$ if it is a
$j$-allowed split of $D$ for some $j\in[k]$.
\end{definition}

The objective of the vector splitting game is to recursively apply allowed splits to a given vector $D\in\R_{\ge0}^{2 \times k}$ with the target of maximizing the sum of the \[p(D') := \max\limits_{x\in\{0,1\}}\bigg(\sum\limits_{j\in[k]}D'_{x,j}-\max\limits_{j\in[k]}D'_{x,j}\bigg)\] values of the resulting vectors (note that when $D'$ is a distribution, then $p(D')$ is the $0$-bit success probability of $D'$ as argued in Section~\ref{sec:intuition}). More precisely, an $n$-round play of the vector splitting game is described by a binary tree of height at most $n$, where the nodes are labeled with \emph{game positions} in $\R_{\ge 0}^{2 \times k}$. The root is labeled with the initial position $D$. For each non-leaf $v$, the labels of the two children form an allowed split of the label of $v$. The payoff of such a play is $\sum_{D'} p(D')$, where $D'$ runs over all leaves of the game tree. The aim is to maximize the payoff. Right from this definition, it is clear that the maximum payoff achievable in an $n$-round game started in position $D$, the \emph{value} of this game, is $\succ_n(D)$ as defined below.

\begin{definition}\label{def:succn}
For all  $n\in\N$  and for all $D \in \R_{\ge0}^{2\times k}$, we recursively define
\begin{enumerate}
\item[(i)] $\succ_{0}(D) := \max\limits_{x\in\{0,1\}}\bigg(\sum\limits_{j\in[k]}D_{x,j}-\max\limits_{j\in[k]}D_{x,j}\bigg)$;
\item[(ii)] $\succ_{n}(D) := \max\limits_{(D_{0},D_{1})} \bigg( \succ_{n-1}(D_{0})+\succ_{n-1}(D_{1})\bigg)$, if $n\ge1$. Here the maximum is taken over all allowed splits $(D_{0},D_{1})$ of $D$.
\end{enumerate}
\end{definition}

For an example of an admissible game, we refer to Figure~\ref{fig:twobittree} in Section~\ref{sec:best-so-far}.

It is easy to see that the game values are non-decreasing in the number of rounds, but bounded. The limiting value is thus well-defined.

\begin{lemma}\label{lem:limsup}
Let $D\in\R_{\ge0}^{2 \times k}$ and $n \in \N$. Then $\succ_n(D) \le \|D\|_1$ and $\succ_{n+1}(D) \ge \succ_n(D)$. Consequently, $\succ(D) := \lim_{n \to \infty} \succ_n(D)$ is well-defined and is equal to $\sup_{n \in \N} \succ_n(D)$.
\end{lemma}
\begin{proof}
  The previous definition and an elementary induction shows $\succ_n(D) \le \|D\|_1$. Since $(D,0)$ is an allowed split of $D$ and $\succ_n(0) = 0$ by the previous observation, we have $\succ_{n+1}(D) \ge \succ_n(D) + \succ_n(0) = \succ_n(D)$. 
\end{proof}

\begin{proposition}[scalability]
\label{prop:scaling}
Let $D \in \R_{\ge 0}^{2 \times k}$ and $\lambda \ge 0$. Then $\succ_n(\lambda D) = \lambda \, \succ_n(D)$ for all $n \in \N$. Consequently, $\succ(\lambda D) = \lambda \, \succ(D)$. 
\end{proposition}
\begin{proof}
  The statements follow right from the definition of $\succ_n$ and $\succ$ via induction.
\end{proof}

\begin{proposition}[monotonicity]
\label{prop:mon}
Let $D, E \in \R_{\ge 0}^{2 \times k}$ with $E \ge D$ (component-wise). Then $\succ_n(E) \ge \succ_n(D)$ for all $n \in \N$. Consequently, $\succ(E) \ge \succ(D)$.
\end{proposition}
\begin{proof}
Clearly $\succ_{0}(E)\ge\succ_{0}(D)$. 
Hence assume that for some $n\in\N$, we have $\succ_{n}(E)\ge\succ_{n}(D)$
for all $E\ge D$. Let $(D^0,D^1)$ with $D^0=(d^0_{ij})$ and $D^1=D-D^0=(d^1_{ij})$ be an allowed split of $D = (d_{ij})$. Building on these, define $E^0 := E^0(D^0,D) = (e^0_{ij})$ with $e^0_{ij} := d^0_{ij} \frac{e_{ij}}{d_{ij}}$ and $E^1 := E^1(D^1,D) := E - E^0$, hence $e^1_{ij} = d^1_{ij} \frac{e_{ij}}{d_{ij}}$. Then $(E^0,E^1)$ is an allowed split of $E$ satisfying $E^0 \ge D^0$ and $E^1 \ge D^1$. Hence  
  \begin{align*}
  \succ_{n+1}(E) &\ge \max_{(D^0,D^1)} ( \succ_{n}(E^0)+\succ_{n}(E^1)) \\
  & \ge \max_{(D^0,D^1)} (\succ_{n}(D^0)+\succ_{n}(D^1)) =\succ_{n+1}(D),
  \end{align*}
where the maxima range over all allowed splits $(D^0,D^1)$
of $D$.
\end{proof}

From the previous definitions and observations, we derive that the game values for games starting with a distribution $D$, that is, $\|D\|_1 = 1$, and the success probabilities of the optimal cryptogenographic protocols for $D$, are equal.

\begin{lemma}
\label{lem:lim-n-bit} Let $D\in\R_{\ge0}^{2 \times k}$ with $\|D\|_1=1$. Then for all $n \in \N$, our definitions of $\succ_n$ coincide with the ones of Brody et al., which are the success probabilities of the best $n$-round cryptogenographic protocols for the distribution $D$. Consequently, also the definition of $\succ(D)$ coincides.
\end{lemma}
\begin{proof}
  Let us for the moment denote the success probabilities defined by Brody et al.\ by $s_n(D)$ and $s(D)$ and then show that $\succ_n(D) = s_n(D)$ and consequently $\succ(D) = s(D)$ for all distributions $D$. 
  
  By definition, we have $\succ_0(D) = s_0(D)$ for all distributions $D$. Lemma 4.1 and 4.2 of~\cite{BrodyJSW14} establish that the first round of any cryptogenographic protocol for the distribution $D$ with some probability $\lambda$ leads to a distribution $D^0$ and with probability $\bar \lambda := 1-\lambda$ leads to a position $D^1$ such that $D = \lambda D^0 + \bar\lambda D^1$ and $D^0, D^1$ are proportional to $D$ on $\{0,1\} \times ([k] \setminus \{j\})$ for some $j \in [k]$. Conversely, for any such $\lambda, D^0, D^1$ there is a one-round cryptogenographic protocol leading to the distribution $D^0$ with probability $\lambda$ and to $D^1$ with probability $\bar\lambda$. Hence for any $n\ge 1$, the success probability $s_n(D)$ of the optimal $n$-round protocol for the distribution $D$ is $s_n(D) = \max_{\lambda, D^0, D^1} (\lambda s_{n-1}(D^0) + \bar\lambda s_{n-1}(D^1))$, where $\lambda, D^0, D^1$ run over all values as above. Note that these are exactly those values which make $(\lambda D^0, \bar\lambda D^1)$ an allowed split of $D$. By induction and scalability, we obtain 
  \begin{align*}
  s_n(D) &= \max_{\lambda, D^0, D^1} (\lambda \, \succ_{n-1}(D^0) + \bar\lambda \, \succ_{n-1}(D^1)) \\
  & = \max_{\lambda, D^0, D^1} (\succ_{n-1}(\lambda D^0) + \succ_{n-1}(\bar\lambda D^1))  \\
  & = \max_{(\bar D^0, \bar D^1)} (\succ_{n-1}(\bar D^0) + \succ_{n-1}(\bar D^1)) = \succ_n(D),
  \end{align*}
  where the last maximum is taken over all allowed splits $(\bar D^0, \bar D^1)$ of $D$.
\end{proof}

\subsection{Example: The Best-so-far 2-Player Protocol}
\label{sec:best-so-far}

We now turn to the case of two players. We use this subsection to describe the best known protocol for two players in the different languages. We also use this mini-example to sketch the approaches used in the following subsections to design superior protocols. 

For two players, we usually write a game position $D = (d_{01}, d_{11}, d_{02}, d_{12}) \in \R_{\ge 0}^{2 \times 2}$ as $D = (a,b,c,d)$. The $0$-round game value (equaling the success probability of the $0$-bit protocol) then is \[\succ_0(D) =\max\{\min\{a,c\},\min\{b,d\}\}.\] 

As a warmup, let us describe the best known 2-player protocol $\TwoBit$ in the two languages. In the language of Brody et al., the first player can send a (randomized) bit that transforms the initial distribution $(\tfrac 14,\tfrac 14,\tfrac 14,\tfrac 14)$ with probability $\tfrac 12$ each into the distributions $(\tfrac 13,\tfrac 13,\tfrac 16,\tfrac 16)$ and $(\tfrac 16,\tfrac 16,\tfrac 13,\tfrac 13)$. In the first case, the second player can send a bit leading to each of the distributions $(\tfrac 13,\tfrac 13,\tfrac 13,0)$ and $(\tfrac 13,\tfrac 13,0,\tfrac 13)$ with probability $\tfrac 12$, both having a $0$-bit success probability of $\tfrac 13$. In the second possible result of the first move, the first player can lead to an analogous situation. Consequently, after two rounds of the protocol we end up with four equally likely distributions all having a $0$-bit success probability of $\tfrac 13$. Hence the protocol $\TwoBit$ has a success probability of $\tfrac 13$.

In the language of the splitting games, we can forget about the probabilities and simply split up the initial distribution. Using the scaling invariance, to ease reading we scaled up all numbers by a factor of $12$. Figure~\ref{fig:twobittree} shows the game tree corresponding to the $\TwoBit$ protocol. It shows that $\succ_2(3,3,3,3) \ge 4$, proving again the existence of a cryptogenographic protocol for the distribution $(\tfrac 14,\tfrac 14,\tfrac 14,\tfrac 14) = \tfrac 1{12} (3,3,3,3)$ with success probability $\tfrac 4 {12} = \tfrac 13$.

Note that each allowed split $(D^0,D^1)$ of $D$ implies the inequality $\succ(D) \ge \succ(D^0) + \succ(D^1)$, which follows from clause (ii) of Definition~\ref{def:succn} and taking the limit $n \to \infty$. Hence the game tree giving the $\tfrac 13$ lower bound for the success probability equivalently gives the following proof via inequalities.
\begin{eqnarray*}
\succ(3,3,3,3) & \ge & \succ(2,2,1,1)+\succ(1,1,2,2),\\
\succ(2,2,1,1)=\succ(1,1,2,2) & \ge & \succ(1,0,1,1)+\succ(0,1,1,1),\\
\succ(1,0,1,1)=\succ(0,1,1,1) & \ge & \succ_{0}(0,1,1,1)=1.
\end{eqnarray*}

\begin{figure}[h]
\centering
\includegraphics[width=0.6\textwidth]{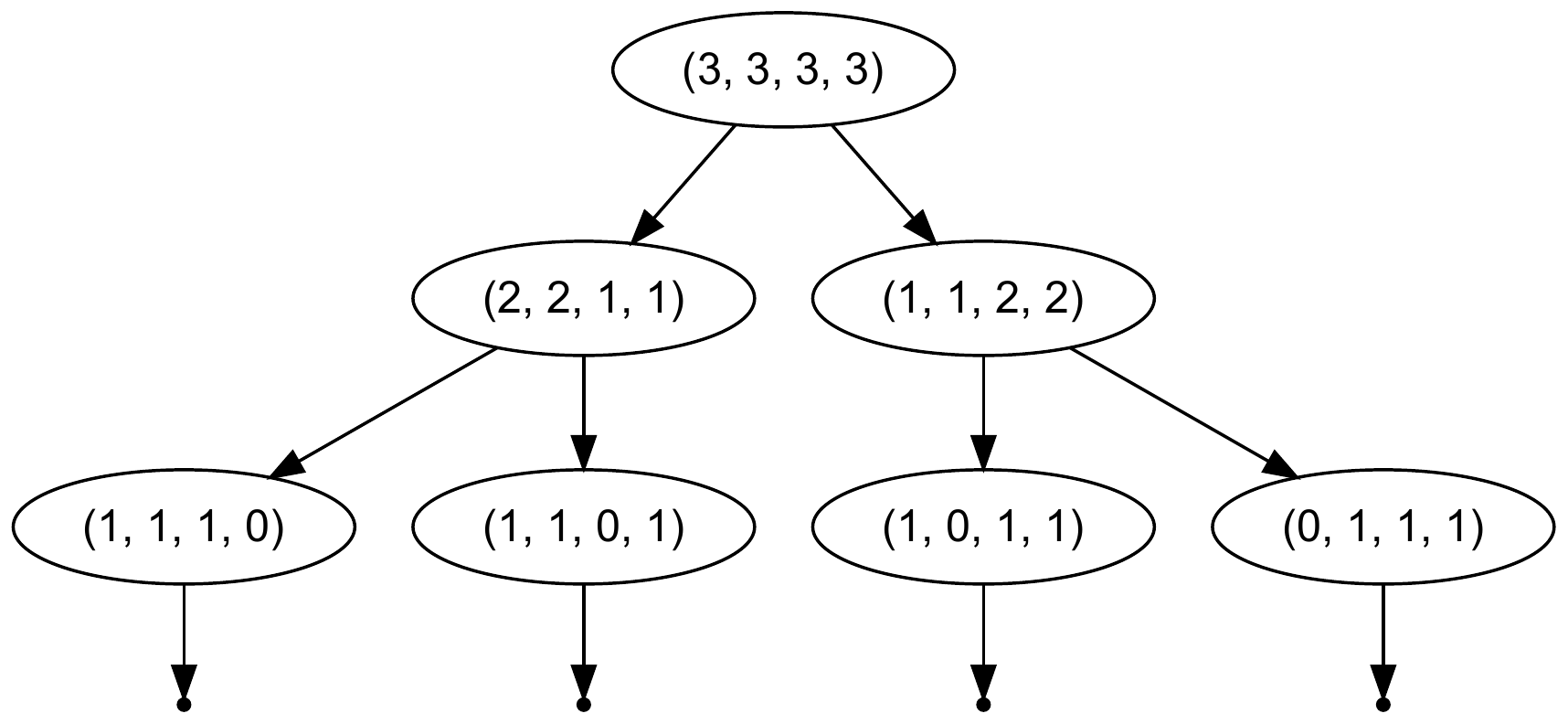}
\caption{Game tree corresponding to $\TwoBit$}
\label{fig:twobittree}
\end{figure}

The splitting game and the inequality view will in the following be used to design stronger protocols (better lower bounds for the optimal success probability). We shall compute good game trees by computing lower bounds for the game values of a discrete set of positions via repeatedly trying allowed splits. For example, the above game tree for the starting position $(3,3,3,3)$ could have easily be found by recursively computing the game values for all positions in $\{0,1,2,3\}^4$. 

It turns out that such an automated search leads to better results when we also allow scaling moves (referring to Proposition~\ref{prop:scaling}). For example, in the above mini-example of computing optimal game values for all positions $\{0,1,2,3\}^4$, we could try to exploit the fact that $\succ(1,1,1,1) = \tfrac 13 \succ(3,3,3,3)$. Such scaling moves are a cheap way of working in $\{0,1,2,3\}^4$, but trying to gain the power of working in $\{0,1,\ldots, 9\}^4$, which is computationally more costly, especially with regard to memory usage. Scaling moves may lead to repeated visits of the same position, resulting in cyclic structures. Here translating the allowed splits used in the tree into the inequality formulation and then using an LP-solver is an interesting approach (detailed in Section~\ref{sec:verification}). It allows to post-optimize the game trees found, in particular, by solving cyclic dependencies. This leads to slightly better game values and  compacter representations of game trees. 

\subsection{Useful Facts}
\label{sec:usefulFacts}

For some positions of the vector splitting game, the true value is easy to determine. We do this here to later ease the presentation of the protocols.

\begin{restatable}{proposition}{concUB}
\label{prop:concUB}We have $\succ(a,b,c,d)\le\min\{a,c\}+\min\{b,d\}$.\end{restatable}

This statement is a simple corollary of the concavity method detailed in Section~\ref{sec:concavityMethod}, hence at this point, we only state the proposition and postpone the proof.

\begin{proposition}\label{prop:zero}
Let $D=(a,b,c,0)$. Then $\succ(D)=\succ_{0}(D)=\min\{a,c\}$.
\end{proposition}
\begin{proof}
Clearly, $\succ(D)\ge\succ_{0}(D)=\min\{a,c\}$. By Proposition~\ref{prop:concUB},
we obtain $\succ(D)\le\min\{a,c\}$, proving the claim.
\end{proof}

\begin{proposition}\label{prop:large}
  If $D = (a,b,c,d)$ is such that $a+b \le \min\{c,d\}$, then $\succ(D) = a+b$.
\end{proposition}
\begin{proof}
  We have $D \ge D' := (a,b,a+b,a+b)$ and $(D^0,D^1)$ with $D^0 = (a,0,a,a)$ and $D^1 = (0,b,b,b)$ is an allowed split of $D'$. Hence $\succ(D) \ge \succ(D') \ge \succ(D^0) + \succ(D^1) = a+b$. The upper bound follows immediately from Proposition~\ref{prop:concUB}.
\end{proof}

\subsection{Small Protocols Beating the $1/3$ Barrier}
\label{sec:simpleBetterProtocol}

We now present a sequence of protocols showing that there are cryptogenographic protocols having a success probability strictly larger than $\tfrac 13$. These protocols are still relatively simple, so we also obtain a human-readable proof of the following result.

\begin{theorem}\label{thm:simpleExnon-optimal}
  $\succ(\frac{1}{4},\frac{1}{4},\frac{1}{4},\frac{1}{4}) \ge \frac{449}{1334} \approx 0.3341$. 
\end{theorem}

\begin{proof}
To be able to give a readable mathematical proof, we argue via inequalities for game values $\succ(\cdot)$. We later discuss how the corresponding protocols (game trees) look like. 

We first observe the following inequalities, always stemming from allowed splits (the
underlined entries are proportional). Whenever Proposition~\ref{prop:zero} or~\ref{prop:large} determine a value, we exploit this without further notice.
\begin{eqnarray*}
  \succ(\underline{12,12},12,12) & \ge & \succ(\underline{7,7},6,4)+\succ(\underline{5,5},6,8),\\
  \succ(\underline{5,5},6,8) & \ge & \succ(\underline{2,2},0,2)+\succ(\underline{3,3},6,6) = 2+6 = 8.
\end{eqnarray*}
This proves
$\succ(12,12,12,12)\ge8+\succ(7,7,6,4)$.
To analyze $\succ(7,7,6,4)$, we use the allowed split
\begin{eqnarray}
  \succ(7,7,\underline{6,4}) & \ge & \succ(4,5,\underline{3,2}) + \succ(3,2,\underline{3,2}) \label{eq:succ7764}
\end{eqnarray}
and regard the two positions $(4,5,3,2)$ and $(3,2,3,2)$ separately in some detail.

\textit{Claim 1: The value of $(4,5,3,2)$ satisfies $\succ(4,5,3,2) \ge \tfrac{55}{12}$.}
By scaling, we have $\succ(4,5,3,2)=\frac{1}{2}\,\succ(8,10,6,4)$.
We present the allowed splits
\begin{eqnarray*}
\succ(\underline{8,10},6,4) & \ge & \succ(\underline{4,5},2,4)+\succ(\underline{4,5},4,0) = \succ(4,5,2,4)+4,\\
\succ(4,5,\underline{2,4}) & \ge & \succ(1,2,\underline{1,2})+\succ(3,3,\underline{1,2}) = \succ(1,2,1,2)+3,
\end{eqnarray*}
hence $\succ(8,10,6,4)\ge\succ(1,2,1,2)+7.$ To bound the latter
term, we use the scaling $\succ(1,2,1,2)=\frac{1}{6}\,\succ(6,12,6,12)$
and consider the allowed splits
\begin{eqnarray*}
\succ(\underline{6,12},6,12) & \ge & \succ(\underline{5,10},3,9)+\succ(\underline{1,2},3,3) = \succ(5,10,3,9)+3,\\
\succ(5,10,\underline{3,9}) & \ge & \succ(0,6,\underline{2,6})+\succ(5,4,\underline{1,3}) = 6+4 = 10.
\end{eqnarray*}
Thus $\succ(6,12,6,12) = 13$ and $\succ(1,2,1,2)=\frac{13}{6}$. This shows
$\succ(4,5,3,2)=\frac{1}{2}\,(\succ(1,2,1,2)+7)=\frac{55}{12}$.

\textit{Claim 2: We have $\succ(3,2,3,2) \ge \tfrac 53 + \tfrac 29 \, \succ(7,7,6,4)$.}
By scaling, we obtain $\succ(3,2,3,2)=\frac{1}{3}\,\succ(9,6,9,6)$ and compute
\begin{eqnarray*}
\succ(9,6,\underline{9,6}) & \ge & \succ(6,3,\underline{6,4})+\succ(3,3,\underline{3,2}),\\
\succ(6,3,\underline{6,4}) & \ge & \succ(3,0,\underline{3,2})+\succ(3,3,\underline{3,2}) = 3 + \succ(3,3,3,2),
\end{eqnarray*}
and hence  $\succ(9,6,9,6)\ge 3 + 2\,\succ(3,3,3,2)$. To bound the latter term, we scale $\succ(3,3,3,2)=\frac{1}{3}\,\succ(9,9,9,6)$
and present the allowed splits
\begin{eqnarray*}
\succ(\underline{9,9},9,6) & \ge & \succ(\underline{7,7},6,4)+\succ(\underline{2,2},3,2),\\
\succ(\underline{2,2},3,2) & \ge & \succ(\underline{1,1},1,0)+\succ(\underline{1,1},2,2) = 1+2 = 3.
\end{eqnarray*}
Thus $\succ(3,3,3,2)\ge 1+\tfrac{1}{3}\,\succ(7,7,6,4)$, implying $\succ(3,2,3,2) = \tfrac 53 + \tfrac 29\,\succ(7,7,6,4)$.

\textit{Putting things together.}
Claims 1 and 2 together with (\ref{eq:succ7764}) give 
\[\succ(7,7,6,4) \ge \tfrac{75}{12} + \tfrac 29\,\succ(7,7,6,4) .\] 
Solving this elementary equation, we obtain $\succ(7,7,6,4) \ge \tfrac{225}{28}$, and consequently, $\succ(12,12,12,12) \ge \tfrac {225}{28} + 8 = \tfrac{449}{28} = 16 + \tfrac 1 {28}$. Scaling leads to the claim of the theorem $\succ(\tfrac14,\tfrac14,\tfrac14,\tfrac14) = \tfrac 13 + \tfrac 1 {1344} = \frac{449}{1344}=0.33407738\dots$. 
\end{proof}

\begin{figure}[h]
\centering
\includegraphics[width=\textwidth]{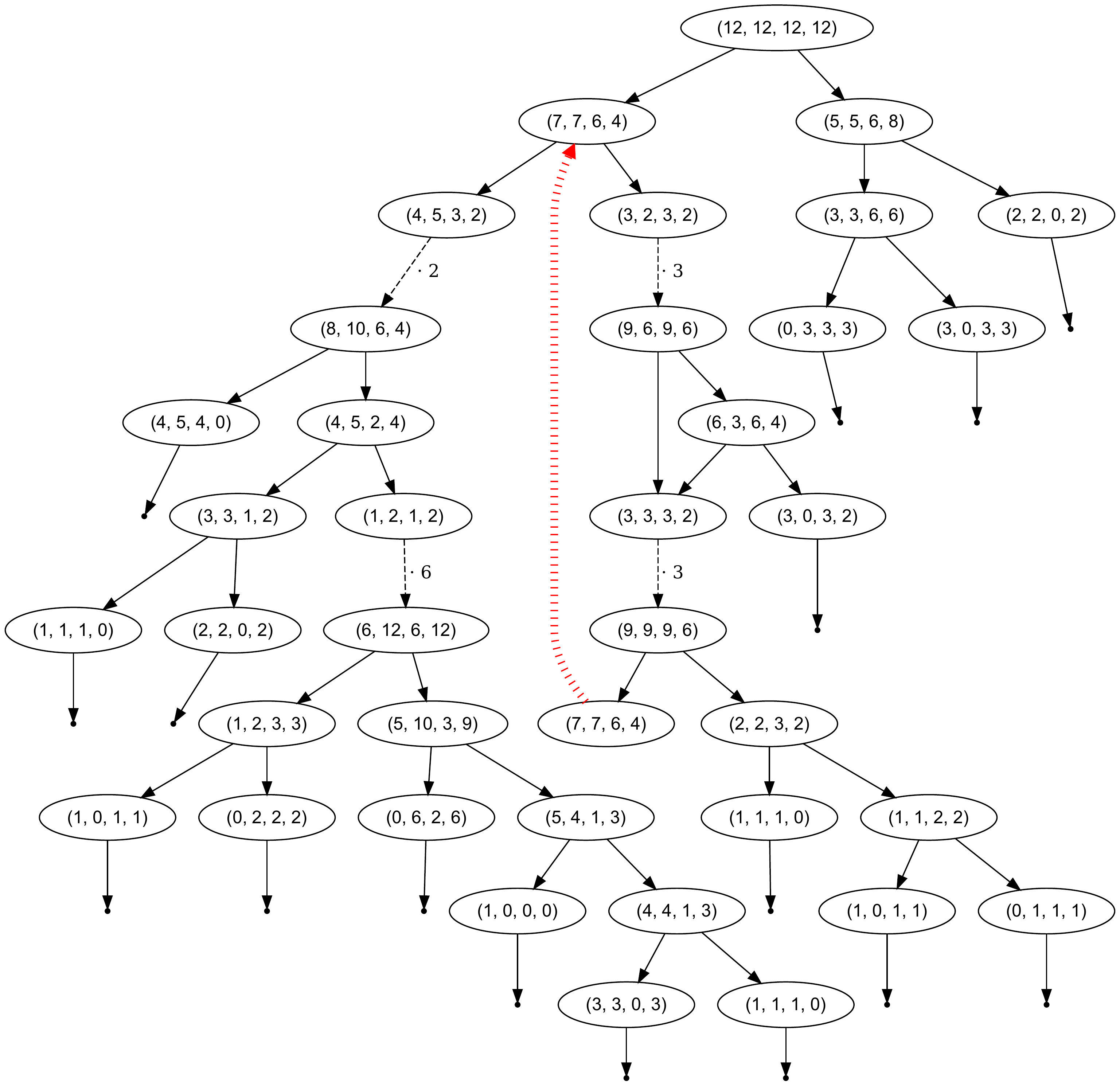}
\caption{Game tree representation of the protocols of Theorem~\ref{thm:simpleExnon-optimal}}
\label{fig:tree}
\end{figure}

When translating the inequalities into a game tree (see Figure~\ref{fig:tree} for the result), we first observe that in Claim 2 we obtained two different nodes labeled with the position $(3,3,3,2)$. Since there is no reason to treat them differently, we can identify these two nodes and thus obtain a more compact representation of the game tree. This is the reason why the node labeled $(3,3,3,2)$ in Figure~\ref{fig:tree} has two incoming edges.

Interestingly, such identifications can lead to cycles. If we translate the equations for position $(7,7,6,4)$ and its children into a graph, then we observe that the node for $(7,7,6,4)$ has a descendant also labeled $(7,7,6,4)$ (this is what led to the inequality $\succ(7,7,6,4) \ge \tfrac{75}{12} + \tfrac 29\,\succ(7,7,6,4)$). By transforming this inequality to $\succ(7,7,6,4) \ge \tfrac{225}{28}$, we obtain a statement that is true, but that does not anymore refer to an actual (finite) game tree. However, there is a sequence of game trees with values converging to the value we determined. These trees are obtained from recursively applying the above splitting procedure for $(7,7,6,4)$ a certain number $\ell$ of times and then using the $0$-round tree for the lowest node labeled $(7,7,6,4)$. The value of this game tree is $8 + \sum_{i = 0}^{\ell-1} (\tfrac 29)^i \tfrac{75}{12} + (\tfrac29)^\ell \, \succ_0(7,7,6,4) = 8 + \tfrac{75}{12}\frac{1-(\frac29)^\ell}{1-\frac29} + 6 \cdot(\tfrac29)^\ell = \tfrac{449}{28} - \tfrac{57}{28} \, (\tfrac 29)^\ell$. Hence for $\ell \ge 3$, this is more than 16 (which represents a success probability of $\tfrac 13$), corresponding to a game tree of height $4+4\ell \ge 16$.

\subsection{Automated Search}
\label{sec:search}

The vector splitting game formulation enables us to search for good cryptogenographic protocols as follows. We try to determine the game values of all positions from a discrete set $\disc:=\{0,\dots,T\}^{2 \times k}$ by repeatedly applying allowed splits. More precisely, we store a function
$s:\disc\to\R$ that gives a lower bound on the game value $\succ(D)$ of each position $D \in \disc$. We initialize this function with $s\equiv\succ_{0}$ and then in order of ascending $\|D\|_1$ try all allowed splits $D = D_0 + D_1$ and update $s(D) \leftarrow s(D_0) + s(D_1)$ in case we find that $s(D)$ was smaller. 

Recall that for any secret distribution $D$, the game value $\succ(D)$ is the supremum success probability of cryptogenographic protocols for $D$. Hence, e.g., the value $s(T,\ldots,T) / (2Tk) \le \succ(1/(2k),\ldots,1/(2k))$ is a lower bound for this success probability. As we will discuss later, by keeping track of the update operations performed, we can not only compute such a lower bound, but also concrete protocols.

Since even for $k=2$, the size of the position space $\disc$ and the number of allowed splits increase quickly
with $T$, only moderate choices of $T$ are computationally feasible,
limiting the power of this approach drastically. Surprisingly, introducing
a simple scaling step is sufficient to overcome this problem and enables
us to find protocols that are better than the previous-best
protocol $\TwoBit$. Algorithm~\ref{alg:search} outlines our basic search procedure.

\begin{algorithm}
\begin{algorithmic}[1] 
\State{Initialize $s\equiv \succ_0$}
\Repeat
\State{\textit{Splitting Step:}}
\ForAll{$D\in \disc$ in ascending order of $\|D\|_1$}
\ForAll{relaxed splits $(D_0,D_1)\in \disc^2$ of $D$}
\If{$s(D)< s(D_0)+s(D_1)$}
\State{Update $s(D) \gets s(D_0)+s(D_1)$}
\EndIf
\EndFor
\EndFor
\State{\textit{Scaling Step:}}
\ForAll{$D \in \disc$ and all $\lambda\in\N$ with $\lambda D\in \disc$}
\If{$s(D) < \frac{s(\lambda D)}{\lambda}$}
\State{Update $s(D) \gets \frac{s(\lambda D)}{\lambda}$}
\EndIf
\EndFor
\Until{termination criterion is met}
\end{algorithmic}

\caption{Computing a lower bound on $\succ$ for the discrete set of positions
$\disc=\{0,\dots,T\}^{2 \times k}$.}\label{alg:search}
\end{algorithm}

Instead of restricting to optimize only over all allowed splits $(D_0, D_1)\in \disc^2$ of $D\in \disc$, however, we use monotonicity of the discretization to exploit even more reasonable splits. For ease of presentation, we focus here on the 2-player case (the generalization to larger values of $k$ is straightforward).
\begin{definition}
\label{def:relaxed}
We call distributions $D_0 = (a_0,b_0,c_0,d_0)\in \disc$ and $D_1=(a_1,b_1,c_1,d_1)\in \disc$ a \emph{relaxed split} of $D=(a,b,c,d)\in \disc$, if $a = a_0 + a_1$, $b=b_0+b_1$, $c=c_0+c_1$, $d_0 = \lfloor  (d/c) \cdot c_0\rfloor$ and $d_1 = \lfloor (d/c) \cdot c_1 \rfloor$.
\end{definition}
\begin{observation}\label{obs:relaxed}
Any relaxed split $(D_0,D_1)\in \disc^2$ of $D\in \disc$ yields a feasible lower bound $\succ(D) \ge \succ(D_0)+\succ(D_1)$.
\end{observation}
\begin{proof}
The statement follows from noting that $\bar{D_0} := (a_0,b_0,c_0,(d/c)\cdot c_0)$ and $\bar{D_1} := (a_1,b_1,c_1, (d/c) \cdot c_1)$ yield an allowed split of $D$. Thus $\succ(D) \ge \succ(\bar{D_0}) + \succ(\bar{D_1}) \ge \succ(D_0) + \succ(D_1)$, where the last inequality follows from monotonicity (Proposition~\ref{prop:mon}).
\end{proof}

Note that while the definition relaxes an allowed split only at coordinate $d$, by symmetry of $\succ$, we obtain the same lower bound when relaxing at other coordinates. This allows us to even split distributions $D=(a,b,c,d)\in \disc$ where neither of the ratios $a/b$ and $c/d$  occur ``perfectly'' in another distribution $D'\in \disc$ by dismissing some ``vector mass'', i.e., rounding down from $d/c \cdot c_i$ to $\lfloor d/c \cdot c_i \rfloor$. Although these splits might appear inherently wasteful (as this loss can never be regained), the best protocols that we find do indeed make use of (a small number of) such relaxed splits. 

\paragraph{Implementation Details.}
Since $\succ$ is symmetric in both the secret and the player dimension, we may
assume the following standard form for distributions $D=((D[0,1],D[1,1]),\dots,(D[0,k],D[1,k]))\in\disc$
to speed up computation and reduce memory usage: $D[0,1]  \ge D[1,1]$ and $D[0,i]  \ge D[0,i+1]\quad\text{for }i=1,\dots,k-1$. More specifically, for $k=2$, we can without loss of generality even assume that $a\ge b,c,d$.

In principle, an implementation should take care to avoid propagation of floating-point rounding errors, since previously computed entries are reused heavily and identical values are regularly recalculated in a number of different ways. Instead of using interval arithmetic, however, we chose to use a simple, fast implementation ignoring potential rounding errors: This is justified by (i) our LP-based post-optimization which gives a proof of the obtained lower bound (that can be checked using an exact LP solver), hence correctness of the output remains certified, and (ii) the fact that our discretization of the search space introduces an inherent imprecision that very likely dominates the floating-point rounding errors.

The probably most desirable termination criterion is to run the search until no improvements can be found. However, when the running time becomes a bottleneck issue, we can restrict the search to a small fixed number of iterations. This is especially useful in combination with the post-optimization (as the gain in the result per iteration is decreasing for later iterations of the process, which intuitively give increasingly accurate approximations of infinite ``ideal'' protocols -- the post-optimization could potentially resolve these cyclic structures earlier in the process).

\paragraph{Results.}

The success probabilities of the protocols computed following the above approach, using different values for $T$, are given in the first line of Table~\ref{table:results}. Further results exploiting the post-optimization are given in Table~\ref{table:resultsPost} in Section~\ref{sec:main}.

\begin{table}
\begin{tabular}{l|lllll}
$T$ & 15 & 20  &  25 & 30\\
\hline
Automated search & 0.3369432925 & 0.3376146092 &  0.3379027186 & 0.3381689066 \\
Iterations & 119 & 129 & 141 & 146 \\
Constraints & 535 & 1756  & 4217 & 13958 \\
Game Positions & 394 & 1326  & 2956 & 9646 
\end{tabular}
\caption{Lower bounds $s(T, \ldots, T)/(4T)$ on $\succ(\frac{1}{4},\frac{1}{4},\frac{1}{4},\frac{1}{4})$ stemming from the automated search (line 1). Given are also the number of iterations until the automated search procedure converged, i.e., stopped finding improvements using relaxed splits or scalings, and the number of game positions and constraints that had an influence on the value of $s(T, \ldots, T)$.}
\label{table:results}
\end{table}

\subsection{Post-Optimization via Linear Programming}
\label{sec:verification}

When letting Algorithm~\ref{alg:search} also keep track of at what time which update operation was performed, this data can be used to extract strategies for the splitting game (and cryptogenographic protocols). Some care has to be taken to only extract those intermediate positions that had an influence on the final game value for the position we are interested in (see below). 

While this approach does deliver good cryptogenographic protocols, manually verifying the correctness of the updates or analyzing the structure of the underlying protocol quickly becomes a difficult task, as the size of the protocol grows rapidly.

Fortunately, it is possible to output a compact, machine-verifiable certificate for the lower bound obtained by the automated search that might even prove a better lower bound than computed: Each update step in the automated search corresponds to a valid inequality of the form $\succ(D)\ge \succ(D_0) + \succ(D_1)$, $\succ(D)\ge \succ_0(D)$ or $\lambda \cdot \succ(D) = \succ(\lambda \cdot D)$. We can extract the (sparse) set $\inequ(T,T,T,T)$ of those inequalities that lead to the computed lower bound on $\succ(T,T,T,T)$.

\paragraph{Reconstructing the Strategy.}
\label{sec:buildProof}

Memorizing the best splits found by the dynamic programming updates, it is straightforward to reconstruct the best strategy found by the automated search. For preciseness, we define the \emph{$i$-th update step} for $i=2k-1$ as the $k$-th splitting step (during the execution of Algorithm~\ref{alg:search}) and for $i=2k$ as the $k$-th scaling step, i.e., each update step is alternatingly a splitting and a scaling step. For every distribution $D$ and update step $i$, we maintain an index $L(D,i)$ defined as the last update step before and including $i$ in which $s(D)$ has been updated to a better value, or 0 if $s(D)$ has never been updated. Moreover, for every $D$ and update step $i$ in which $s(D)$ has been updated, we keep a constraint $I(D,i)$ which represents the inequality or equality used to update $s(D)$ to its best value in update step $i$. More specifically, $I(D,i)$ stores the inequality $\succ(D) \ge \succ(D_0) + \succ(D_1)$ if the update step $i$ represents a splitting step of $D$ into $(D_0,D_1)$ and  $\lambda \succ(D) = \succ(\lambda D)$ if the update step $i$ represents a scaling of $D$ by $\lambda$. Furthermore, we let $I(D,0)$ represent the inequality $\succ(D) \ge \succ_0(D)$. This gives rise to the procedure $\inequ(D,i)$ (see Algorithm~\ref{alg:ineqs}) that returns a (sparse) list of inequalities proving the lower bound computed by the automated search (after letting it run for $i$ iterations). 

\begin{algorithm}
\begin{algorithmic}[1] 
\Function{$\inequ$}{$D,i$}
\If{$L(D,i) < i$}
\State \Return $\inequ(D,L(D,i))$
\EndIf
\If {$i$ is a splitting step with $I(D,i) = "\succ(D) \ge \succ(D_0)+\succ(D_1)"$} 
  \State \Return $I(D,i) \cup \inequ(D_0,i) \cup \inequ(D_1,i)$
\ElsIf {$i$ is a scaling step with $I(D,i) = "\lambda \succ(D) =  \succ(\lambda D)"$}  
  \State \Return $I(D,i) \cup \inequ(\lambda D,i-1)$ 
\ElsIf {$i=0$ with $I(D,0) = "\succ(D) \ge  \succ_0(D)"$}
  \State \Return{$I(D,0)$}
\EndIf
\EndFunction
\end{algorithmic}
\caption{Computing the proof $\inequ(D,I)$ of the lower bound on $\succ(D)$ obtained by running the automated search for $I$ iterations.}
\label{alg:ineqs}
\end{algorithm}

\paragraph{The Linear Program.}
\label{sec:buildLP}

Consider replacing each occurrence of $\succ(D')$ in the set of inequalities $\inequ(T,T,T,T)$ found by the automated search by a variable $s_{D'}$. We obtain a system of linear inequalities $S$ that has the feasible solution $s_{D'} = \succ(D')$ (for every occurring vector $D'$). Hence in particular, the optimal solution of the linear program of \emph{minimizing $s_D$ subject to $S$} is a lower bound on $\succ(D)$. It is easy to see that this solution is at least as good as the solution stemming from the automated search alone. It can, however, even be better, in particular when a game strategy yields cyclic visits to certain positions.  Table~\ref{table:resultsPost} contains, for different values of $T$, the success probabilities found by automated search (run with a bounded iteration number of 20) and by this above linear programming approach. The table also contains the number of linear inequalities (and game positions) that were extracted from the automated search run. We observe that consistently the LP-based solution is minimally better. We also observe that the number of constraints is still moderate, posing no difficulties for ordinary LP solvers (which stands in stark contrast to feeding all relaxed splits and scalings over the complete discretization to the LP solver, which quickly becomes infeasible). 

Hence the advantage of our approach of extracting the constraints from the automated search stage is that it generates a much sparser sets of constraints that still are sufficiently meaningful. After solving the LP,  we can further sparsify this set of inequalities by deleting all inequalities that are not tight in the optimal solution of the LP, since these cannot correspond to the best splits found for the corresponding vector $D$, yielding a smaller set of relevant inequalities, which might help to analyze the structure of strong protocols.

\subsection{Our Best Protocol}
\label{sec:main}

We report the best protocol we have found using the approach outlined in the previous sections.

\begin{theorem}
In the 2-player cryptogenography problem, $\succ(\frac{1}{4},\frac{1}{4},\frac{1}{4},\frac{1}{4})\ge 0.3384736$.
\end{theorem}
\begin{proof}
On \verb|http://people.mpi-inf.mpg.de/~marvin/verify.html|, we provide a linear program based on feasible inequalities on the discretization $\disc$ with $T=50$. To verify the result, one only has to (1) check validity of each inequality, i.e., checking whether each constraints encodes a feasible scaling, relaxed split or zero-bit success probability and (2) solve the linear program. Since we represent the distributions $D=(a,b,c,d)$ using a normal form $a\ge b,c,d$ (to break symmetries), checking validity of each splitting constraint is not completely trivial, but easy. We provide a simple checker program to verify validity of the constraints. The LP is output in a format compatible with the LP solver \texttt{lp\_solve}\footnote{\texttt{lpsolve.sourceforge.net}}.
\end{proof}

\begin{table}
\begin{tabular}{l|llllll}
$T$ & 30 & 35  &  40 & 45 & 50\\
\hline
Automated search & 0.3381086510 & 0.3381937725 & 0.3383218072 & 0.3383946540  & 0.3384414508 \\
LP solution & 0.3381527322 & 0.3382301900 & 0.3383547901 & 0.3384303130 & 0.3384736461 \\
Iterations & 20 & 20 & 20 & 20 & 20  \\
Constraints & 5373 & 8882 & 12410 & 18659 & 24483 \\
Game Positions & 4126 & 6789 & 9396 & 13992 & 18248
\end{tabular}
\caption{Lower bounds $s(T, \ldots, T)/(4T)$ on $\succ(\frac{1}{4},\frac{1}{4},\frac{1}{4},\frac{1}{4})$ stemming from the automated search only (line 1) and from the LP solution of the linear system extracted from the automated search data, when the number of iterations is restricted to 20.} 
\label{table:resultsPost}
\end{table}

\section{A Stronger Hardness Result}\label{sec:2playerHardness}

In this section, we prove that any $2$-player cryptogenographic protocol has a success probability of at most $0.3672$. This improves over the previous $0.375$ bound of~\cite{BrodyJSW14}.

\begin{theorem}\label{thm:hardness}
We have $\succ(\frac{1}{4},\frac{1}{4},\frac{1}{4},\frac{1}{4}) \le \frac{47}{128} = 0.367188.$
\end{theorem}

To prove the result, we apply the concavity method used by Brody et al.~\cite{BrodyJSW14} which consists of finding a function $s$ that (i) is lower bounded by $\succ_0$ for all distributions and (ii) satisfies a certain concavity condition. We first relax the lower bound requirement to hold only for six particular simple distributions (namely $(1,0,0,0),\dots,(0,0,0,1),(\frac{1}{2},0,\frac{1}{2},0)$ and $(0,\frac{1}{2},0,\frac{1}{2})$) instead of all distributions. This simplifies the search for a suitable stronger candidate function satisfying (i) - it remains to verify condition (ii) for the thus found candidate function. 

We first describe the previously used concavity method in Section~\ref{sec:concavityMethod} and apply it to our new upper bound function in Section~\ref{sec:adapted}.

\subsection{Revisiting the concavity method}
\label{sec:concavityMethod}

We revisit the original result of Brody et al.~\cite[Theorem 4.4]{BrodyJSW14} establishing the concavity method for the cryptogenography problem. Similar concavity arguments have appeared earlier in different settings in information complexity and information theory (e.g.,~\cite{BravermanGPW13}).
In what follows, we mostly use the language of convex combinations of secret distributions in allowed planes. To this aim, let $\dists := \{ D \in \R^{2\times 2}_{\ge 0} \mid \|D\|_1 = 1 \}$ denote the set of these distributions. For a given distribution $D = (a,b,c,d) \in \dists$, there are two \emph{allowed planes} through $D$, namely $\{(a',b',\delta c, \delta d) \mid a', b', \delta \in \R\} \cap \dists$ and $\{(\delta a, \delta b, c', d') \mid c', d', \delta \in \R\} \cap \dists$. 

\begin{lemma}[Concavity Method, {\cite[Theorem 4.4]{BrodyJSW14}}]%\cite[Theorem 4.4]{BrodyJSW14}]
\label{lem:oldConcavity}Let $s:\dists\to\R$ satisfy
\begin{enumerate}
\item[(C1)] $s(D)\ge \lambda s(D_{0})+ (1-\lambda) s(D_{1})$ for all $\lambda \in [0,1]$ and all $D_0, D_1$ such that $D = \lambda D_0 + (1-\lambda) D_1$ and $D_0, D_1$ lie in the same allowed plane through $D$.
\item[(C2)] $s(D)\ge \succ_0(D)$ for all $D \in \dists$.
\end{enumerate}
Then $\succ(D)\le s(D)$ holds for all $D\in\dists$.
\end{lemma}

To transfer this method to the vector splitting formulation, consider an $s':\R_{\ge 0}^{2\times 2}\to \R$, which is scalable (i.e., $s'(\lambda D)=\lambda s'(D)$ for all $\lambda \in \R_{\ge 0}, D \in \R_{\ge 0}^{2\times 2}$) and when restricted to $\dists$ equals $s$. Then condition (C1) is equivalent to $s'(D) \ge s'(D_{0}) + s'(D_1)$ for all allowed splits $(D_0,D_1)$ of $D$, i.e., superadditivity of $s'$ for all allowed splits. Condition (C2) remains effectively the same: $s'(D)\ge \succ_0(D)$ for all $D\in \R^{2\times 2}_{\ge 0}$. If these conditions hold, we obtain $\succ(D) \le s'(D)$ for all $D\in \R^{2\times 2}_{\ge 0}$, which is equivalent to $\succ(D) \le s(D)$ for all $D\in \dists$.

As a simple application of the concavity method, we can now give the simple proof of Proposition~\ref{prop:concUB} stated in Section~\ref{sec:usefulFacts}.

\concUB*
\begin{proof}
We make use of the vector splitting formulation. Define  $s_{UB}(a,b,c,d):= \min\{a,c\}+\min\{b,d\}$. We have \[\succ_{0}(a,b,c,d)=\max\{\min\{a,c\},\min\{b,d\}\}\le \min\{a,c\}+\min\{b,d\} = s_{UB}(a,b,c,d),\]
which proves condition (C2) of Lemma~\ref{lem:concavity}.
Note that $f:(x,y)\mapsto\min\{x,y\}$ is superadditive and hence
$s_{UB}$, as a sum of superadditive functions,  is superadditive
as well. This proves $s_{UB}(D)\ge s_{UB}(D_{0})+s_{UB}(D_{1})$ even for
\emph{all} splits $D=D_{0}+D_{1}$ (not only allowed splits).
\end{proof}

The following lemma is an extension of the concavity theorem. We relax the condition that $s$ is lower bounded by $\succ_0$ on \emph{all} distributions to now on only six particular, very simple distributions.

\begin{lemma}[Concavity Method, adapted]
\label{lem:concavity}Let $s:\dists\to\R$ satisfy
\begin{enumerate}
\item[(C1)] $s(D)\ge \lambda s(D_{0})+ (1-\lambda) s(D_{1})$ for all $\lambda \in [0,1]$ and all $D_0, D_1$ such that $D = \lambda D_0 + (1-\lambda) D_1$ and $D_0, D_1$ lie in the same allowed plane through $D$.
\item[(C2')] 
\renewcommand\labelitemi{\tiny$\bullet$}
\begin{itemize}
\item $s(1/2,0,1/2,0),s(0,1/2,0,1/2) \ge 1/2$, and
\item $s(1,0,0,0), s(0,1,0,0), s(0,0,1,0),s(0,0,0,1) \ge 0$.
\end{itemize}
\end{enumerate}
Then $\succ(D)\le s(D)$ holds for all $D\in\dists$.
\end{lemma}

\begin{proof}
To appeal to Lemma~\ref{lem:oldConcavity}, we need to show that (C1) and (C2') imply (C2), i.e., for all $D=(a,b,c,d) \in \dists$, we have $s(D)\ge \succ_0(D) = \max\{\min\{a,c\},\min\{b,d\} \}$. 

To prove this statement, we again  find it more convenient to use the vector splitting formulation. To this aim, we extend $s$ to $s': \R^{2\times 2}_{\ge 0} \to \R$ by defining $s'(D) := \|D\|_1\cdot s(\frac{D}{\|D\|_1})$. Recall that in this view (C1) is equivalent to $s'(D) \ge s'(D_0) + s'(D_1)$ for all allowed splits $(D_0,D_1)$ of $D$. Assume that $\min \{a,c\} \ge \min\{b,d\}$ (the other case is symmetric). Consider $D_0 := (a,0,c,d)$, $D_1 := (0,b,0,0)$, which is a $1$-allowed split of $D$. By scaling and (C2'), we have $s'(D_1)=b\cdot s(0,1,0,0)\ge 0$. We split $D_0$ into $E_{0}:=(a,0,c,0)$ and $E_{1}:=(0,0,0,d)$, which is a $2$-allowed split. Again, by scaling and (C2'), we obtain $s'(E_1)\ge 0$. Assuming that $a\ge b$ (since the other case is symmetric), we finally split $E_0$ into $F_0 := (b,0,b,0)$ and $F_1:=(a-b,0,0,0)$, which is a $1$-allowed split that satisfies, by scaling and (C2'), $s'(F_1)\ge 0$ and $s'(F_0) = 2b \cdot s(\frac{1}{2},0,\frac{1}{2},0)\ge b$. Thus, $s(D)=s'(D)\ge b = \succ_0(D)$.
\end{proof}

\subsection{The adapted upper bound function}
\label{sec:adapted}

Motivated by our relaxation of the condition of the concavity method, we adapt the upper bound function of Brody et al.~\cite{BrodyJSW14} and set
\begin{eqnarray*}
s(a,b,c,d) & := & \frac{1-f(a,b,c,d)}{4},\\
f(a,b,c,d) & := & a^2 + b^2 + c^2 + d^2 - 6(ac+bd) + 8abcd.
\end{eqnarray*}
Note that we have changed the upper bound function by introducing an additional term of $8abcd$, which attains a value of zero on the distributions $(\frac{1}{2},0,\frac{1}{2},0)$, $(1,0,0,0)$, etc., thus not affecting the zero-bit success probability condition of the concavity method. Additionally, this function also satisfies the concavity condition, which we will prove later.

\begin{lemma}
\label{lem:sIsConcave}
The thus defined $s$ satisfies the concavity condition (C1) of Lemma~\ref{lem:concavity}.
\end{lemma}

As an immediate consequence of the concavity method of Lemma~\ref{lem:concavity}, we obtain the upper bound of $\succ(\frac{1}{4},\frac{1}{4},\frac{1}{4},\frac{1}{4}) \le \frac{47}{128}$ of Theorem~\ref{thm:hardness}.

\begin{proof}[Proof of Theorem~\ref{thm:hardness}]
Simple calculations show that $s$ satisfies $s(1,0,0,0)=0$ and $s(1/2,0,1/2,0)=1/2$, which by symmetry implies that $s$ satisfies condition (C2'). Since additionally condition (C1) is satisfied by Lemma~\ref{lem:sIsConcave}, the concavity method (Lemma~\ref{lem:concavity}) is applicable and yields $\succ(\frac{1}{4},\frac{1}{4},\frac{1}{4},\frac{1}{4}) \le s(\frac{1}{4},\frac{1}{4},\frac{1}{4},\frac{1}{4}) = \frac{47}{128}$.
\end{proof}

In the remainder of this section, we will show that $s$ does not violate concavity on allowed planes (condition (C1)), i.e., prove Lemma~\ref{lem:sIsConcave}. To this end,  we exclusively use the convex combination view, in which $\|D_0\|_1 = \|D_1\|_1=\|D\| = 1$.

For $q\ge 0$, we define $f_q : S \to \R$, where $S:=\{(a,b)\mid 0\le a,b \le 1, a+b\le 1\}$, by
\[
f_q(a,b):= f\left(a,b,\frac{1}{1+q}(1-a-b),\frac{q}{1+q}(1-a-b)\right),
\]
i.e., for fixed $q\ge 0$, this function maps each pair $(a,b)\in S$ to the value of $f$ at the unique distribution $(a',b',c',d')$ with $a'=a$, $b'=b$ and $d'/c' = q$.
\begin{lemma}
\label{lem:ifconcave}
If $f_q$ is convex for all $q\ge 0$, then $s$ satisfies the concavity condition (C1) of Lemma~\ref{lem:concavity}.
\end{lemma}
\begin{proof}
 Let $D = (a,b,c,d)$ be an arbitrary distribution and let $D_0=(a_0,b_0,c_0,d_0), D_1=(a_1,b_1,c_1,d_1)$ be such that $D = \lambda D_0 + (1-\lambda) D_1$ for some $\lambda \in [0,1]$ and $D_0,D_1$ lie in the same allowed plane through $D$.  By the symmetry $s(a,b,c,d)=s(c,d,a,b)$, we can without loss of generality assume that $D_0,D_1$ and $D$ are proportional on the entries of player~$2$ (i.e., a 1-allowed split). 

Assume first that $c=d=0$, then any $D_0$ and $D_1$ lie on the line $(t,1-t,0,0)$ with $t\in \R$. Define the restriction of $s$ to this line as $\tilde{s}(t) := s(t,1-t,0,0) = \frac{1}{4} \left(1-(1-t)^2-t^2\right)$ and note that $\tilde{s}(t)$ is concave by $\frac{d^2 \tilde{s}}{d t^2} = -1$. Thus (C1) is satisfied for $c=d=0$.

Hence we may assume $c > 0$ or $d>0$ and more specifically, by the symmetry $s(a,b,c,d)=s(b,a,d,c)$, that $c>0$. We set  $q:=  \frac{d}{c}$ and observe that $q = \frac{d_0}{c_0} = \frac{d_1}{c_1}$, since $(D_0,D_1)$ is a 1-allowed split of $D$. Then by $\|D\|_1 = \|D_i\|_1 = 1$, we have $D=(a,b,\frac{1}{1+q}(1-a-b),\frac{q}{1+q}(1-a-b))$ and $D_i = (a_i,b_i, \frac{1}{1+q}(1-a_i-b_i), \frac{q}{1+q}(1-a_i-b_i))$ (for $i\in \{0,1\}$). Recall that $D = \lambda D_0 + (1-\lambda) D_1$. The claim now follows from the simple calculation
\begin{align*}
s(D) & = \frac{1-f(a,b,\frac{1}{1+q}(1-a-b),\frac{q}{1+q}(1-a-b))}{4}  = \frac{1-f_q(a,b)}{4}  \\
& \ge \lambda \frac{1- f_q(a_0,b_0)}{4} + (1-\lambda)\frac{1-f_q(a_1,b_1)}{4} \qquad \text{[by convexity of } f_q\text{]}\\
& = \lambda s(D_0) + (1-\lambda) s(D_1).
\end{align*}
\end{proof}

It remains to analyze the concavity of $f_q$.

 \begin{lemma}
\label{lem:fqConcave}
For all $q\ge 0$, $f_q$ is concave on $S$.
\end{lemma}
\begin{proof}
We use the fact that $f_q$ has continuous second partial derivatives and thus is convex if the Hessian $H:=\Hess(f_q)$ is positive semidefinite. By straight-forward calculations, we obtain
\[
\Hess(f_q) =
\left(
\begin{array}{cc}
 \frac{4 \left(q^2+4 \left(2 b^2+(3 a-2) b+1\right) q+4\right)}{(q+1)^2} & \frac{4 \left(2 q^2+\left(6 a^2+8 (2 b-1) a+6 b^2-8 b+5\right) q+2\right)}{(q+1)^2} \\
 \frac{4 \left(2 q^2+\left(6 a^2+8 (2 b-1) a+6 b^2-8 b+5\right) q+2\right)}{(q+1)^2} & \frac{4 \left(4 q^2+4 \left(2 a^2+(3 b-2) a+1\right) q+1\right)}{(q+1)^2} \\
\end{array}
\right).
\]
To verify positive semidefiniteness, we exploit the criterion that $\Hess(f_q)$ is positive semidefinite if its principal minors are non-negative, i.e.,
\begin{eqnarray}
\frac{4 \left(4+4 \left(1+(-2+3 a) b+2 b^2\right) q+q^2\right)}{(1+q)^2} & \ge & 0, \qquad \text{ and } \label{eq:H1}\\
\Det[H] & \ge & 0. \label{eq:H2}
\end{eqnarray}
Verifying~\eqref{eq:H1} is equivalent to showing that $4+4 p_1(a,b) q+q^2\ge 0$, where $p_1(a,b):= 1+(-2+3 a) b +2 b^2$. Note that by $0\le a,b \le 1$, we have $p_1(a,b) \ge 1-2b + 2b^2 \ge 1- 2b \ge -1$. Thus,  
\[ 4+4 p_1(a,b) q+q^2\ge 4 - 4q  +q^2 = (q-2)^2 \ge 0,\]
proving~\eqref{eq:H1}.

Regarding~\eqref{eq:H2}, straight-forward calculations reveal that
\[ \Det[H] =  \frac{64q}{(1+q)^4} \left(p_2(a,b) + p_3(a,b)q + p_4(a,b)q^2\right), \]
where
\begin{eqnarray*}
p_2(a,b) & = & 2 a^2+6 b-a b-4 b^2, \\
p_3(a,b) & = & 12 (a+b) - 23 (a^2 + b^2) - 32 ab + 24 (a^3 (1-b) + b^3 (1-a))  \\
& & + 48 (ab^2 + a^2 b)  - 30 a^2 b^2 - 9 (a^4 + b^4),\\ 
p_4(a,b) & = & 6 a-4 a^2-a b+2 b^2.
\end{eqnarray*}

\begin{claim}
For all $0\le a,b \le 1$ with $a+b\le 1$ we have 
\[p_2(a,b),p_3(a,b),p_4(a,b)\ge 0.\]
\end{claim}
\begin{proof}
Note that by $a,b\le 1$, we have $p_2(a,b) = 2a^2 + 6b - ab - 4b^2 \ge 2a^2 + 6b - b - 4b = 2a^2 + b\ge 0$. Since $p_4(a,b) = p_2(b,a)$, it directly follows that  $p_4(a,b) \ge 0$. 

To prove the remaining statement $p_3(a,b) \ge 0$, we will exploit the following basic inequalities
\begin{eqnarray}
a^4 + b^4 & \le & a^3 (1-b) + b^3 (1-a), \label{eq:simplepot4} \\
a^2 b^2 & \le & \frac{1}{2}(ab^2 + a^2 b  - a b^3 - a^3 b),\label{eq:simplepot2prod}
\end{eqnarray}
which directly follow from plugging in $a\le 1-b$ and $b\le 1-a$ into the left-hand sides. We compute
\begin{align*}
p_3(a,b) & = & & 12 (a+b) - 23 (a^2 + b^2) - 32 ab + 24 (a^3 (1-b) + b^3 (1-a)) & &  \\
& & & + 48 (ab^2 + a^2 b) - 30 a^2 b^2  - 9 (a^4 + b^4), & & \\ 
& \ge & & 12 (a+b) - 23 (a^2 + b^2) - 32 ab + 15 (a^3 (1-b) + b^3 (1-a)) & &  \\
& & & + 48 (ab^2 + a^2 b) - 30 a^2 b^2 & & \text{[by \eqref{eq:simplepot4}]} \\
& \ge & & 12 (a+b) - 23 (a^2 + b^2) - 32 ab + 15 (a^3 + b^3) & & \\
& & & + 33 (ab^2 + a^2 b) & & \text{[by \eqref{eq:simplepot2prod}]}\\
& = & & 12 (a+b) - 23 (a+b)^2 + 14 ab + 4 (a^3 + b^3) & & \\
& & & + 11 (a+b)^3 & & \\
& \ge & & 12 (a+b) - 23 (a+b)^2 + 11 (a+b)^3,  
\end{align*}
Basic calculus shows that $t: [0,1]\to \R, s\mapsto 12 s - 23 s^2 + 11s^3$ has global minima $t(0) = t(1) = 0$ and thus $p_3(a,b) \ge t(a+b) \ge 0$.
\end{proof}

Thus $p_2(a,b) + p_3(a,b)q + p_4(a,b)q^2 \ge 0$ follows from $q\ge 0$, which implies that $\Det[H]\ge 0$.
Hence, we have verified~\eqref{eq:H1} and \eqref{eq:H2}, which proves that $f_q$ is convex for all $q\ge 0$. 
\end{proof}

\begin{proof}[Proof of Lemma~\ref{lem:sIsConcave}]
Combining Lemmas~\ref{lem:ifconcave} and~\ref{lem:fqConcave} yields that $s$ is concave on all allowed planes.
\end{proof}

\section{Conclusion}

Despite the fundamental understanding of the cryptogenography problem obtained by Brody et al.~\cite{BrodyJSW14}, determining the success probability even of the 2-player case remains an intriguing open problem. The previous best protocol with success probability $1/3$, while surprising and unexpected at first, is natural and very symmetric (in particular when viewed in the convex combination or vector splitting game formulation). We disprove the hope that it is an optimal protocol by exhibiting less intuitive and less symmetric protocols having success probabilities up to $0.3384$. Concerning hardness results, our upper bound of $0.3671875$ shows that also the previous upper bound of $3/8$ was not the final answer. These findings add to the impression that the cryptography problem offers a more complex nature than its simple description might suggest and that understanding the structure of good protocols is highly non-trivial. 

We are optimistic that our methods support a further development of improved protocols and bounds. (1) Trivially, investing more computational power or optimizing the automated search might lead to finding better protocols. (2) Our improved protocols might motivate to (manually) find infinite protocol families exploiting implicit properties and structure of these protocols. (3) Our reformulations, e.g., as vector splitting game, might ease further searches for better protocols and for better candidate functions for a hardness proof.

\bibliography{cryptogeno}

\begin{thebibliography}{1}

\bibitem{BravermanGPW13}
Mark Braverman, Ankit Garg, Denis Pankratov, and Omri Weinstein.
\newblock From information to exact communication.
\newblock In {\em Symposium on Theory of Computing (STOC'13)}, pages 151--160.
  {ACM}, 2013.

\bibitem{BrodyJSW14}
Joshua Brody, Sune~K. Jakobsen, Dominik Scheder, and Peter Winkler.
\newblock Cryptogenography.
\newblock In {\em Innovations in Theoretical Computer Science (ITCS'14)}, pages
  13--22. {ACM}, 2014.

\bibitem{Jakobsen14}
Sune~K. Jakobsen.
\newblock Information theoretical cryptogenography.
\newblock In {\em 41st International Colloquium on Automata, Languages, and
  Programming (ICALP'14)}, volume 8572 of {\em Lecture Notes in Computer
  Science}, pages 676--688. Springer, 2014.

\bibitem{JakobsenO16}
Sune~K. Jakobsen and Claudio Orlandi.
\newblock How to bootstrap anonymous communication.
\newblock In {\em Innovations in Theoretical Computer Science (ITCS'16)}, pages
  333--344, 2016.

\bibitem{Matousek01}
Ji\v{r}{\'{\i}} Matou\v{s}ek.
\newblock On directional convexity.
\newblock {\em Discrete {\&} Computational Geometry}, 25:389--403, 2001.

\end{thebibliography}

\end{document}